\documentclass[11pt]{article}
\usepackage{amsthm, amsfonts,mathrsfs,amssymb,clrscode,amsmath}
\usepackage{mathptmx}
\usepackage{quantph}  
\newcommand{\abs}[1]{\vert #1 \vert}
\newcommand{\gen}[1]{\langle #1 \rangle}
\newcommand{\ket}[1]{\vert #1 \rangle}

\newcommand{\Int}{\mathbb Z}

\newcommand{\bound}[2]{\in\{ #1,\ldots,#2\}}
\newcommand{\Sp}{\mathscr{S}}

\newcommand{\poly}{\textrm{poly}}

\newcommand{\Dd}{\mathscr{D}}

\newcommand{\field}{\mathbb{F}}
\newcommand{\fieldK}{\mathbb{K}}

\newcommand{\st}{\:\vert\:}
\newcommand{\triv}{\{e\}}
\long\def\symbolfootnote[#1]#2{\begingroup%
\def\thefootnote{\fnsymbol{footnote}}\footnote[#1]{#2}\endgroup}
\newtheorem{theorem}{Theorem}[section]
\newtheorem{proposition}{Proposition}[section]
\newtheorem{definition}{Definition}[section]

\newtheorem{lemma}{Lemma}[section]
\begin{document}
\begin{center} 

{\LARGE \bf An Efficient Quantum Algorithm for some Instances\\ of the Group Isomorphism Problem}\vspace{5mm}

{\Large Fran{\c c}ois Le Gall \symbolfootnote[0]{This work was done while the author was a researcher at Kyoto University, affiliated with the ERATO-SORST Quantum Computation and Information Project, Japan Science and Technology Agency.}}\vspace{3mm}

{\it Department of Computer Science}\\
{\it Graduate School of Information Science and Technology}\\
{\it The University of Tokyo}

email: legall@is.s.u-tokyo.ac.jp \vspace{1mm}

\setlength{\baselineskip}{10pt}
      \begin{quotation}
\noindent{\bf Abstract.}\hbox to 0.5\parindent{}
 In this paper we consider the problem of testing whether two finite groups are isomorphic.
Whereas the case where both groups are abelian is well understood and can be solved 
efficiently, very little is known about the complexity of isomorphism testing for nonabelian 
groups. Le Gall has  constructed an efficient classical algorithm for a class of groups 
corresponding to one of the most natural ways of constructing nonabelian groups from abelian groups: the groups that are 
extensions of an abelian group $A$ by a cyclic group $\Int_m$ with the order of $A$ coprime with $m$.
More precisely, the running time of that algorithm is almost linear in the order of the input groups.
In this paper we present a \emph{quantum} algorithm solving the same problem in time polynomial 
in the \emph{logarithm} of the order of the input groups. 
This algorithm works in the black-box setting and is 
the first quantum algorithm solving instances of the nonabelian group isomorphism problem exponentially faster than the best known classical algorithms.
\end{quotation}
\setlength{\baselineskip}{11pt}
      \begin{quotation}
\end{quotation}
\end{center} 
\vspace{-10mm}

\section{Introduction}
\paragraph{Background}

Testing group isomorphism (the problem asking to decide, for two given finite groups
$G$ and $H$,  whether there exists an isomorphism between $G$ and
$H$)
is a fundamental problem in computational group theory but little is known about its complexity. 
It is known that the group isomorphism problem (for groups given by their multiplication tables) reduces to the graph isomorphism problem \cite{Kobler+93}, 
and thus the group isomorphism problem is in the complexity class $NP\cap coAM$  (since the graph isomorphism problem is in this class  \cite{BabaiSTOC85}).
Miller \cite{MillerSTOC78} has developed a general technique to check group isomorphism in time $O(n^{\log n+O(1)})$, where $n$ denotes the size of the input groups and Lipton, 
Snyder and Zalcstein \cite{Lipton+76} have given an algorithm working in $O(\log^2{n})$ space. However, no polynomial time algorithm is known for the general case of this problem.

Another line of research is the design of algorithms solving the group isomorphism problem for particular classes of groups.
For abelian groups polynomial-time algorithms follow directly from efficient algorithms for the computation of the Smith normal form of integer 
matrices  \cite{Chou+SICOMP82, Kannan+SICOMP79}. More efficient methods have been given by Vikas \cite{Vikas96} and 
Kavitha \cite{KavithaJCSS07} for abelian groups given by their multiplication tables, and fast parallel algorithms have been constructed by McKenzie and Cook \cite{McKenzie+SICOMP87}
for abelian permutation groups.
The current fastest algorithm solving the abelian group isomorphism problem for groups given as black-boxes has been 
developed by Buchmann and Schmidt \cite{Buchmann+05} and works in time $O(n^{1/2}(\log n)^{O(1)})$.
However, as far as nonabelian groups are concerned, very little is known.
For solvable groups Arvind and Tor{\'a}n \cite{Arvind+CCC04} have shown that the group isomorphism problem is in $NP\cap coNP$ 
under certain complexity assumptions
but, until recently, the only polynomial-time algorithms testing isomorphism of nontrivial classes of nonabelian groups
were a result by Garzon and Zalcstein \cite{Garzon+JCSS91}, which holds for a very restricted class, and a 
body of works initiated by Cooperman et al.~\cite{cooperman+97} on simple groups, which will be discussed later.

Very recently, Le Gall \cite{LeGallSTACS09} proposed an efficient classical algorithm solving the group isomorphism problem over another class of 
nonabelian groups.
Since for abelian groups the group isomorphism problem can be solved efficiently, that work focused on one of the most natural 
next targets: cyclic extensions of abelian groups. 
Loosely speaking such extensions are constructed by taking an abelian group $A$ and adding one element $y$ that, in general, does not commute with the elements
in $A$. 
More formally the class of groups considered in \cite{LeGallSTACS09}, denoted by $\Sp$, was the following.
\begin{definition}
Let $G$ be a finite group. 
The group $G$ is said to be in the class $\Sp$ if there exists a normal abelian subgroup $A$ in $G$ and 
an element $y\in G$ of order coprime with $\abs{A}$ such that $G=\gen{A,y}$.
\end{definition}
\noindent In technical words $G$ is an extension of an abelian group $A$ by a  cyclic group $\Int_{m}$ with 
$gcd(\abs{A},m)=1$. 
This class of groups includes all the abelian groups and many non-abelian groups too, as discussed in details in \cite{LeGallSTACS09}.  
For example, for $A=\Int_3^4$ and $m=4$, there are exactly $9$ isomorphism classes in $\Sp$ (1 class of abelian groups and 8 classes of nonabelian groups). 
Moreover, the class $\Sp$ includes several groups that have been the target of quantum algorithms, as discussed later.
The main result in \cite{LeGallSTACS09} was the following theorem.
\begin{theorem}[\cite{LeGallSTACS09}]\label{theorem_old}
There exists a deterministic algorithm checking whether two groups $G$ and $H$ in the class $\Sp$ (given as black-box groups)
are isomorphic and, if this is the case, computing an isomorphism from $G$ to $H$.  
Its running time has for upper bound $n^{1+o(1)}$, where $n=min(\abs{G},\abs{H})$.
\end{theorem}


\paragraph{Statement of our results}
In the present paper, we focus on \emph{quantum algorithms} solving the group isomorphism problem in the black-box setting.
Cheung and Mosca \cite{Cheung+01}
have shown how to compute the decomposition of an abelian group into a direct product of cyclic subgroups in time polynomial in the logarithm of its order on a quantum computer, and thus
how to solve the abelian group isomorphism problem in time polynomial in $\log n$  in the black-box model. 
(Notice that their algorithm is actually a generalization of Shor's algorithm \cite{ShorSICOMP97}, 
which can be seen as solving the group isomorphism problem over cyclic groups.)
This then gives an exponential speed-up with respect to the best known classical
algorithms  for the same task. One can naturally ask whether a similar speed-up can be obtained for classes of nonabelian groups.
In this paper, we prove that this is the case. Our main result is the following theorem.
\begin{theorem}\label{theorem_main}
There exists a quantum algorithm checking with high probability whether two groups $G$ and $H$ in the class $\Sp$ given as black-box groups
are isomorphic and, if this is the case, computing an isomorphism from $G$ to $H$.  
Its running time is polynomial 
in $\log n$, where $n=min(\abs{G},\abs{H})$.
\end{theorem}
To our knowledge, this is the first quantum algorithm solving nonabelian instances of the  group isomorphism problem 
exponentially faster than the best known classical algorithms. 
Our algorithm relies on several new quantum reductions to instances of the so-called abelian Hidden Subgroup Problem, a problem  that can be solved efficiently on a quantum computer. 
Our result can then be seen as an extension of the polynomial time library of computational tasks which can be accomplished
using Shor's factoring and discrete logarithm algorithms \cite{ShorSICOMP97}, and further quantum algorithms for abelian groups. 
We also mention that groups in the class $\Sp$ appear at several occasions in the quantum computation
literature, mostly connected to the Hidden Subgroup Problem over semidirect product groups \cite{Bacon+FOCS05,Ettinger+00, Inui+QIC07,Moore+SODA04}.
Our techniques may have applications in the design of further quantum algorithms for this problem, or for other similar group-theoretic tasks.

\paragraph{Overview of our algorithm}
Our quantum algorithm follows the same line as the classical algorithm in \cite{LeGallSTACS09}, but the two main technical parts
are both significantly improved and modified.

Since a group $G$ in the class $\Sp$ may in general be written as the extension of 
an abelian group $A_1$ by a cyclic group $\Int_{m_1}$ and as the extension of 
an abelian group $A_2$ by a cyclic group $\Int_{m_2}$ with 
$A_1\not\cong A_2$ and $m_1\neq m_2$, we use, as in \cite{LeGallSTACS09}, the concept of a standard
decomposition of $G$, which is an invariant for the groups in the class $\Sp$ in the sense that
two isomorphic groups have similar standard decompositions (but the converse is false). 
A method for computing efficiently standard decompositions in the black-box model was one of the main contributions of \cite{LeGallSTACS09},
where the time complexity of this step was $O(n^{1+o(1)})$ due to the fact that the procedure proposed had to try, in the worst case, 
for each generator $g$ of $G$, all the divisors of $\abs{g}$.
Instead, in the present work we propose a different procedure for this task (Section \ref{section_standard}), which
can be implemented in time polynomial in $\log n$ on a quantum computer,
based on careful reductions to group-theoretic 
problems for which known efficient quantum algorithms are known: order finding, decomposing abelian groups and 
constructive membership in abelian groups.

Knowing standard decompositions of $G$ and $H$ 
allows us to consider only the case where $H$ and $G$ are two extensions of the same abelian group $A$ by the same cyclic group $\Int_m$
(Proposition \ref{proposition_class}).
Two matrices $M_1$ and $M_2$ in the group $GL(r,\field)$ of invertible matrices of size $r\times r$ 
over some well-chosen finite field $\field$ can then be associated to the action of $\Int_m$ on $A$ in the groups $G$ and $H$ respectively. 
The second main technical contribution of \cite{LeGallSTACS09} showed that, loosely speaking, testing isomorphism of $G$ and $H$ then reduces 
(when the order of $A$ is coprime with $m$) to checking whether
there exists an integer $k\bound{1}{m}$ such that $M_1$ and $M_2^k$ are conjugate in $GL(r,\field)$ 
(a precise version of this statement is given in Proposition~\ref{theorem_red} of the present paper).
The strategy adopted in $\cite{LeGallSTACS09}$ 
to solve this problem had time complexity close to $n$ in the worst case (basically, all the integers $k$ in $\{1,\ldots,m\}$ were checked). 
In the present paper, we give a $\poly(\log n)$ time quantum algorithm for this problem.
More generally, we show in Section \ref{section:conjugacy} that the problem of testing, for any two matrices $M_1$ and $M_2$ in $GL(r,\field)$ where $r$ is any positive integer and $\field$ is any finite field,  whether there exists a positive integer $k$ such that $M_1$ and $M_2^k$ are conjugate in the group $GL(r,\field)$ reduces to solving an instance of a problem we call $\proc{Set Discrete Logarithm}$. 
This quantum reduction is efficient in that it can be implemented in time polynomial in both $r$ and $\log\abs{\field}$, and works by considering field extensions of 
$\field$ and matrix invariants of $M_1$ and $M_2$.

Loosely speaking, the problem  $\proc{Set Discrete Logarithm}$ asks, given two sets $\{x_1,\ldots,x_v\}$ and $\{y_1,\ldots,y_v\}$ of elements in $\field$,
to compute an integer $k$ such that $\{y_1^k,\ldots,y_v^k\}=\{x_1,\ldots,x_v\}$, if such an integer exists.
This computational problem is a generalization of the standard discrete logarithm problem (which is basically the case $v=1$) but appears to be much more challenging.\footnote{To illustrate this point, let us consider the following simple strategy: 
for each $j\bound{1}{v}$, try to find some $k$ such that $y_1^k=x_j$ using the quantum algorithm for the standard discrete logarithm problem by Shor \cite{ShorSICOMP97}, and then check whether $\{y_1^k,\ldots,y_v^k\}=\{x_1,\ldots,x_v\}$. 
The problem here is that a $k$ such that  $y_1^k=x_j$ will be only defined modulo $\abs{y_1}$, and it may be the case that $\{y_1^k,\ldots,y_v^k\}\neq\{x_1,\ldots,x_v\}$
but $\{y_1^{k'},\ldots,y_v^{k'}\}=\{x_1,\ldots,x_v\}$ for some $k'$ satisfying $k'=k\bmod \abs{y_1}$. Testing all these $k'$'s can take exponential time.} 
The quantum algorithm we propose (in Section \ref{section:set}) works in time polynomial in $v$ and $\log\abs{\field}$, and relies on a reduction to several instances of the abelian Hidden Subgroup Problem. 
Our solution to the problem $\proc{Set Discrete Logarithm}$ is then an extension of the computational tasks which can be solved efficiently
using known quantum algorithms for abelian groups. 

\paragraph{Other related works}
To our knowledge, the only other work on polylogarithmic time nonabelian group isomorphism testing in the back-box setting is a
body of results, initiated by Cooperman et al.~\cite{cooperman+97}, focusing on identifying simple groups.
Remember that a simple group is a group that has no nontrivial normal subgroup. 
A celebrated result in group theory classifies all the simple finite groups into 26 sporadic groups and a few numbers of infinite classes in which
each group has a label of some prescribed form. 
A natural question that arises is, given a black-box group guaranteed to be simple, how to compute this label, i.e., how to identify this group?
It is known that, based on the mathematical properties of the simple groups, 
it is possible to do this (classically) in polylogarithmic time whenever the input is guaranteed to be a so-called classical group over a field of known characteristic.
We refer to the book by Kantor and Seress \cite{Kantor+01} and references therein for an extensive treatment of this subject.

\section{Preliminaries}\label{section_prelim}
\subsection{Group theory and standard decompositions}
We assume that the reader is familiar with the basic notions of group theory and state without proofs 
definitions and properties of groups we will use in this paper. 

For any positive integer $m$, we denote by $\Int_m$ the additive cyclic group of integers $\{0,\ldots,m-1\}$, and by $\Int_m^\ast$ the multiplicative group of integers in $\{1,\ldots,m-1\}$
coprime with $m$.

Let $G$ be a finite group.
For  any subgroup $H$ and any normal subgroup $K$ of $G$ we denote by $HK$ the subgroup $\{hk\st  h\in H,k\in K\}=\{kh\st  h\in H,k\in K\}$. 
Given a set $S$ of elements of $G$, the subgroup generated by the elements of $S$ is written $\gen{S}$.
We say that two elements $g_1$ and $g_2$ of $G$ are conjugate in $G$ if there exists an element $y\in G$ such that
$g_2=yg_1y^{-1}$. For  any two elements $g,h\in G$ we denote by $[g,h]$ the commutator
of $g$ and $h$, i.e., $[g,h]=ghg^{-1}h^{-1}$. More generally, given two subsets $S_1$ and $S_2$ of $G$, we define
$[S_1,S_2]=\gen{[s_1,s_2]\:|\:s_1\in S_1,s_2\in S_2}$.
The commutator subgroup of $G$ is defined as $G'=[G,G]$. 
The derived series of $G$ is defined recursively as $G^{(0)}=G$ and
$G^{(i+1)}=(G^{(i)})'$. The group $G$ is said to be solvable if there exists some integer $k$ such 
that $G^{(k)}=\triv$.
Given two groups $G_1$ and $G_2$, a map $\phi:G_1\to G_2$ is a homomorphism from $G_1$ to $G_2$ if, for any two elements 
$g$ and $g'$ in $G_1$, the relation $\phi(gg')=\phi(g)\phi(g')$ holds.  We say that $G_1$ and $G_2$ are isomorphic if there exists a 
one-one homomorphism from $G_1$ to $G_2$, and we write $G_1\cong G_2$.

Given any finite group $G$, we denote by $\abs{G}$ its order and, given any element $g$ in $G$, we denote by $\abs{g}$ the order of
$g$ in $G$. For any prime $p$, we say that a group is a $p$-group if its order is a power of $p$.
If $\abs{G}=p_1^{e_i}\ldots p_r^{e_r}$ for distinct prime numbers $p_i$, then 
for each $i\bound{1}{r}$ the group $G$ has a subgroup  of order $p_i^{e_i}$. Such a subgroup is called a Sylow $p_i$-subgroup of $G$.
Moreover, if $G$ is additionally abelian, then each Sylow $p_i$-group is unique and
$G$ is the direct product of its Sylow subgroups.
Abelian $p$-groups have remarkably simple structures: any abelian 
$p$-group is isomorphic to a direct product of cyclic $p$-groups $\Int_{p^{f_1}}\times\cdots\times\Int_{p^{f_s}}$
for some positive integer $s$ and positive integers $f_1\le \ldots\le f_s$, and this decomposition is unique.
We say that a set $\{g_1,\ldots,g_t\}$ of $t$ elements of an abelian group $G$ is a basis of $G$ if 
$G=\gen{g_1}\times\cdots\times\gen{g_t}$ and the order of each $g_i$ is a prime power.

For a given group $G$ in the class $\Sp$ in general many different decompositions as an extension of an abelian group by a 
cyclic group exist. For example, the abelian group $\Int_6=\gen{x_1,x_2 \st x_1^2=x_2^3=[x_1,x_2]=e}$ 
can be written as $\gen{x_1}\times\gen{x_2}$, $\gen{x_2}\times\gen{x_1}$ or $\gen{x_1,x_2}\times\triv$.
That is why we introduce the notion of a standard decomposition, as it was done in $\cite{LeGallSTACS09}$. 
\begin{definition}
Let $G$ be a finite group in the class $\Sp$. 
For any positive integer $m$ denote by $\Dd^m_G$ the set (possibly empty) of pairs $(A,B)$ such that the following
three conditions hold:
(i)
$A$ is a normal abelian subgroup of $G$ of order coprime with $m$; and
 (ii)
 $B$ is a cyclic subgroup of $G$ of order $m$; and
 (iii)
 $G=AB$.
Let $\gamma(G)$ be the smallest positive integer such that $\Dd^{\gamma(G)}_G\neq\varnothing$.
A standard decomposition of $G$ is an element of $\Dd^{\gamma(G)}_G$.
 \end{definition}

\subsection{Black-box groups and the abelian Hidden Subgroup Problem}\label{sub:clBB}

In this paper we work in the black-box model, first introduced (in the classical setting) by  Babai and Szemer{\'e}di \cite{Babai+FOCS84}.
A black-box group is a representation of a group $G$ where elements are represented by strings, and an oracle is available
to perform group operations.
To be able to take advantage of the power of quantum computation when dealing with black-box groups, 
the oracles performing group operations have to be able to deal with quantum superpositions.
These quantum black-box groups have been first studied by Ivanyos et al.~\cite{Ivanyos+03} 
and Watrous \cite{WatrousFOCS00,WatrousSTOC01}, and have become the standard model for studying group-theoretic
problems in the quantum setting.

More precisely, a quantum black-box group is a representation  of a group where elements are represented by 
strings (of the same length, supposed to be logarithmic in the order of the group). 
We assume the usual unique encoding hypothesis, i.e., each element of the group is encoded by a unique string, which is crucial for technical reasons
(without it, most quantum algorithms do not work). 
A quantum oracle $V_G$ is available, such that $V_G(\ket{g}\ket{h})=\ket{g}\ket{gh}$ for any $g$ and $h$ in $G$ (using strings to represent the group elements), 
and behaving in an arbitrary way on other inputs.\footnote{A quantum oracle computing the inverse of elements is not necessary since the inverse of an element can be computed if one knows its order --- 
this latter task can be  done efficiently as stated in Theorem \ref{tasks}.}
We say that a group $G$ is input as a black-box if a set of strings representing generators $\{g_1,\ldots,g_s\}$ of 
$G$ with $s=O(\log\abs{G})$ is given as input, and queries to the oracle can be done at cost 1.
The hypothesis on $s$ is natural since every group $G$ has a generating set of size $O(\log \abs{G})$, and enables us 
to make the exposition of our results easier. Also notice that a set of generators of any size can 
be converted efficiently into a set of generators of size $O(\log\abs{G})$ if randomization 
is allowed \cite{BabaiSTOC91}.

Any efficient quantum black-box algorithm gives rise to an efficient concrete quantum algorithm whenever
the oracle operations can be replaced by efficient procedures. 
Especially, when a mathematical expression of the generators input to the algorithm is known, 
performing group operations can be done directly on the elements in polynomial time (in $\log\abs{G}$) for many
natural groups, including permutation groups and matrix groups. 
This is why the black-box model is one of the most general settings to work with when considering group-theoretic problems, 
and especially when designing  sublinear-time algorithms for such problems.

Quantum algorithms are very efficient for solving computational problems over abelian groups. In the following theorem, we 
describe the main results we will need in this paper.
\begin{theorem}[\cite{Cheung+01,Ivanyos+03,ShorSICOMP97}]\label{tasks}
There exists quantum algorithms solving, in time polynomial in $\log \abs{G}$, the following computational tasks with probability at least $1-1/\poly(\abs{G})$:
\begin{itemize}
\item[(i)]
Given a group $G$ given as a  black-box (with unique encoding) and any element $g\in G$, compute the order of $g$ in $G$.
\item[(ii)]
Given an abelian group $G$ given as a  black-box (with unique encoding), compute a basis $(g_1,\ldots,g_s)$ of $G$.
\item[(iii)] 
Given an abelian group $G$ given as a  black-box (with unique encoding), a basis $(g_1,\ldots,g_s)$ of $G$, and any $g\in G$, 
compute a decomposition of $g$ over $(g_1,\ldots,g_s)$, i.e., integers $u_1,\ldots, u_s$ such that $g=g_1^{u_1}\cdots g_s^{u_s}$.
\end{itemize}
\end{theorem}
More precisely, Task (i) can be solved using a black-box version of Shor's algorithm \cite{ShorSICOMP97}, 
Task (ii) can be solved using Cheung and Mosca's algorithm \cite{Cheung+01}, 
and Task (iii) can be solved using the quantum algorithm by Ivanyos et al.~\cite{Ivanyos+03}.
The discrete logarithm problem is the special case of task (iii) above when $G$ is a cyclic group.
Moreover, since factoring an integer reduces to computing the order of elements in a cyclic group, the efficient solution to Task (iii) implies an efficient solution for 
the integer factoring problem (we refer to Shor's paper \cite{ShorSICOMP97} for a precise description of this reduction).  

Actually, all the tasks in Theorem \ref{tasks} can be seen as black-boxes versions of instances of the so-called Hidden Subgroup Problem (HSP) over abelian groups.
We now recall the definition of this problem, since we will need it in Section \ref{section:conjugacy}.
Let $G$ be a group, $K$ be a subgroup of $G$, and $X$ be a finite set.
A function $f:G\to X$ is said to be $K$-periodic if $f$ is constant on each left coset of $K$, with
distinct value on distinct cosets.
Given as inputs
(i)
a group $G$ given as a set of generators, and
(ii)
a function $f$ given as an oracle,
which is $K$-periodic for an unknown subgroup $K$ of $G$,
the Hidden Subgroup Problem asks to 
output a set of generators for $K$. The abelian Hidden Subgroup Problem 
is the special case where the underlying group $G$ is abelian. 
It is known that the abelian HSP can be solved in time polynomial in $\log \abs{G}$ \cite{Kitaev95}, 
even if $G$ is given as a black-box group with unique encoding~\cite{Ivanyos+03,Mosca99}.

\subsection{Invariant factors and elementary divisors of a matrix}\label{sub:inv}
In this subsection we review the notions of invariant factors and elementary divisors of a matrix.
These are standard results, 
and we refer to any textbook on algebra (e.g., \cite{Dummit+04}) for proofs and more details. 
In this subsection $\field$ denotes a finite field, and $GL(r,\field)$ denotes the group of invertible matrices of size $r\times r$ over $\field$ 
for some positive integer $r$.

Let $a(x)=x^k+b_{k-1}x^{k-1}+\ldots+b_1x+b_0$ be any monic polynomial in $\field[x]$.
The \emph{companion matrix} of $a(x)$, denoted by $C_{a(x)}$ is the $k\times k$ matrix with 1's down the first subdiagonal,
$-b_0$, $-b_1$,\ldots, $-b_{k-1}$ down the last column and zero elsewhere.
For example, the companion matrix of $x^4+b_3x^3+b_{2}x^{2}+b_1x+b_0$ is the matrix
$$
\left(\begin{array}{llll}
0&0&0&-b_0\\
1&0&0&-b_1\\
0&1&0&-b_2\\
0&0&1&-b_{3}
\end{array}\right).
$$

Let $M$ be a matrix in $GL(r,\field)$. Then it is known that there exists a unique list $(a_1(x),\ldots,a_s(x))$ of monic polynomials  in $\field[x]$, with
 each polynomial $a_i(x)$ dividing $a_{i+1}(x)$ for each $i\bound{1}{s-1}$, such that
 $M$ is similar to the block diagonal matrix $diag(C_{a_1(x)},\ldots, C_{a_s(x)})$. This list of polynomials is called the 
 \emph{invariant factors} of $M$, and
this block diagonal matrix is called the \emph{rational normal form} of the matrix $M$ and is unique.
In particular the polynomial $a_s(x)$ is the minimal monic polynomial of $M$, i.e., the (unique) monic polynomial of
smallest degree in $\field[x]$ such that $a_s(M)=0$. 
It is known that matrices are conjugate in $GL(r,\field)$ if and only if they have the same invariant factors (or 
equivalently if they have the same rational normal form).
Moreover, these invariant are the same if $M$ is seen as a matrix over a field extension $\fieldK$ of $\field$, i.e.,
two matrices in $GL(r,\field)$ are similar in $GL(r,\field)$ if and only if  they are similar in $GL(r,\fieldK)$.

Let $\fieldK$ be a field extension of $\field$ that splits the minimal polynomial $a_s(x)$ of $M$, i.e.,
$a_s(x)=(x-\lambda_1)^{b_1}\cdots(x-\lambda_t)^{b_t}$ where the $\lambda_i$'s are
distinct elements of $\fieldK$ and the $b_i$'s are their multiplicities. 
Each invariant factor $a_i(x)$ of $M$ can then be written as $a_i(x)=(x-\lambda_1)^{c_{i1}}\cdots(x-\lambda_t)^{c_{it}}$,
where each $c_{ij}$ is a nonnegative integer in $\{0,\ldots,b_j\}$. 
Then the set of \emph{elementary divisors} of $M$ is the set with possible repetitions
$$
\{
(x-\lambda_j)^{c_{ij}}\:|\:i\bound{1}{s}, j\bound{1}{t} \textrm{ such that }c_{ij}\neq 0
\}.
$$ 
The set of elementary divisors associated to $M$ is unique, and it is known that two matrices are
similar in $GL(r,\field)$ if and only if they have the same set of elementary divisors over $\fieldK$,
when $\fieldK$ is an extension field of $\field$ splitting both their minimal polynomials.
For example, suppose that $r=4$, $s=2$, $a_2(x)=(x-\lambda_1)(x-\lambda_2)^2$,
and $a_1(x)=(x-\lambda_1)$ for distinct elements
$\lambda_1$ and $\lambda_2$ in $\fieldK$. Then the  set of elementary divisors is $\{(x-\lambda_1),(x-\lambda_1),(x-\lambda_2)^2\}$.

The elementary divisors of $M$ are closely connected to the so-called Jordan normal form of $M$.
Let $c$ be a nonnegative integer and $\lambda$ be an element in $\fieldK$.
The Jordan matrix of size $c$ associated to $\lambda$, denoted by $J(\lambda,c)$, 
is the $c\times c$ matrix with $\lambda$ along the main diagonal and 1 along the first superdiagonal.
For example:
$$
J(\lambda,4)=\left(\begin{array}{llll}
\lambda&1&0&0\\
0&\lambda&1&0\\
0&0&\lambda&1\\
0&0&0&\lambda
\end{array}\right).
$$
It is easy to check that the minimal polynomial of $J(\lambda,c)$ is $(x-\lambda)^c$.
In particular, this shows that the set of elementary divisors of $J(\lambda,c)$ is $\{(x-\lambda)^c\}$. 

Suppose that the set of elementary divisors of a matrix $M$ (in $GL(r,\field)$, but seen as a matrix in $GL(r,\fieldK)$ where $\fieldK$ splits its minimal polynomial) is 
$\{
(x-\lambda_k)^{d_{k}}\:|\:k\bound{1}{\ell}
\},$
where the $\lambda_k$'s may not be distinct (and necessarily $r=\sum_{k=1}^\ell d_i$).
Then it is known that $M$ is similar over $GL(r,\fieldK)$ to the block diagonal matrix
$$diag( J(\lambda_1,d_1),\ldots, J(\lambda_\ell,d_\ell)).$$
This block diagonal matrix is called the \emph{Jordan normal form} of $M$ and is unique up to the ordering of the $\lambda_i$'s.
For example the Jordan normal form for the example considered above with the set of elementary divisors $\{(x-\lambda_1),(x-\lambda_1),(x-\lambda_2)^2\}$ 
is
$$
diag(J(\lambda_1,1),J(\lambda_1,1),J(\lambda_2,2))=
\left(\begin{array}{llll}
\lambda_1&0&0&0\\
0&\lambda_1&0&0\\
0&0&\lambda_2&1\\
0&0&0&\lambda_2
\end{array}\right).
$$

\section{Computing a Standard Decomposition}\label{section_standard}
In this section we present a quantum algorithm computing a standard decomposition of any group in the class $\Sp$ in time polynomial in the logarithm of the order of the group.

\subsection{Description of the algorithm}\label{sub_descr}
The precise description of the algorithm, which we denote Procedure $\proc{Decompose}$, is given in metacode in Figure \ref{figure:procedure}.
\begin{figure}[h!]
\hrule\vspace{-3mm}
\begin{codebox}
\Procname{Procedure $\proc{Decompose}$} 
\zi \const{input:} a set of generators $\{g_1,\ldots,g_s\}$ of a group $G$ in $\Sp$ with $s=O(\log\abs{G})$.
\zi \const{output:} a pair $(U,v)$ where $U$ is a subset of $G$ and $v\in G$.
\li compute a set of generators $\{g'_1,\ldots,g'_t\}$ of the derived subgroup $G'$ with $t=O(\log\abs{G})$;
\li compute $\kappa=lcm(\abs{g_1},\ldots,\abs{g_s})$;
\li factorize $\kappa$ and write $\kappa=p_1^{e_1}\cdots p_r^{e_r}$ where the prime numbers $p_i$ are distinct;
\li $U\gets \{g'_1,\ldots,g'_t\}$; $V\gets \varnothing$; $\Sigma\gets\varnothing$;
\li   \For $i =1$ \To $r$
\li \Do 
\li $\Gamma_i\gets \varnothing$;
\li \For $j =1$ \To $s$ \kw{do} $\Gamma_i\gets \Gamma_i\cup\{g_j^{\kappa/p_i^{e_i}}\}$;
\li \If $[\Gamma_i,G'] = e$  \kw{and} $gcd(p_i,\abs{G'})\neq 1$ \kw{then}  $U\gets U\cup \Gamma_i$;
\li \If $[\Gamma_i,G'] = e$  \kw{and} $gcd(p_i,\abs{G'})= 1$
\li\hspace{25mm} \kw{then}  
\li \hspace{30mm}search for an element $\gamma_i\in \Gamma_i$ such that $\gen{\Gamma_i}G'=\gen{\gamma_i,G'}$;
\li \hspace{30mm}\kw{if} no such element exists 
\li \hspace{40mm}\kw{then} $U\gets U\cup \Gamma_i$
\li \hspace{40mm}\kw{else} $\Sigma\gets \Sigma\cup \{\gamma_i\}$; 
\li \hspace{25mm}\kw{endthen} 
\li \If $[\Gamma_i,G']\neq e$  \kw{then}  \{  take an element $\gamma_i\in\Gamma_i$ such that $\abs{\gamma_i}=\max_{\gamma\in\Gamma_i}\abs{\gamma}$;
\li \hspace{32mm} $V\gets V\cup\{\gamma_i\}$;$\:\:$\}
\End
\li \hspace{8mm}\bf{enddo}
\li \For all $w$ in $\Sigma$ 
\li\hspace{6mm}\kw{do} 
\li \hspace{10mm}\kw{if} there exists an element $z$ in $\Sigma$ such that $[w,z]\neq e$
\li \hspace{15mm}\kw{then}
\{ \If $zwz^{-1}\in \gen{w}$ \kw{then} $U\gets U\cup \{w\}$ \kw{else}  $V\gets V\cup \{w\}$; \}
\End
\li \hspace{5mm} \kw{enddo}
\li \For all $w\in\Sigma\backslash(U\cup V)$
\li\hspace{6mm}\kw{do} 
\li \hspace{10mm}\kw{if} $[w,u]=\{e\}$ for all $u\in U$ 
\kw{then} $U\gets U\cup\{w\}$
\kw{else}  $V\gets V\cup\{w\}$;
\li \hspace{5mm} \kw{enddo} 
\li $b\gets \Pi_{g\in V}\abs{g}$; $z\gets \Pi_{g\in V}g$; $v\gets z^{\abs{z}/b}$;
\li output $(U,v)$; 
\end{codebox}\vspace{-3mm}
\hrule
\caption{Procedure $\proc{Decompose}$. }\label{figure:procedure}
\end{figure}
\noindent Further descriptions on how each step is implemented follow.
\begin{itemize}\vspace{-2mm}
\item
At Step 1 a set of generators $\{g'_1,\ldots,g'_{t}\}$ of the derived subgroup $G'$ with $t=O(\log\abs{G})$ is computed in time polynomial
in $\log\abs{G}$ with success probability $1-1/\poly(\abs{G})$ using
the classical algorithm by Babai et al.~\cite{Babai+JCSS95}. 
\item\vspace{-2mm}
The order of $G'$ at Steps 9 and 10, and the orders of elements at Steps 2, 17 and 29 are computed using the quantum algorithms for Tasks (i) and (ii) in Theorem \ref{tasks}.
\item\vspace{-2mm}
The least common multiple at Step 2 is computed using standard algorithms, and is factorized at Step 3 using Shor's factoring algorithm \cite{ShorSICOMP97}.
\item\vspace{-2mm}
At Step 12, notice that $[\Gamma_i,G']=e$ implies that $\gen{\Gamma_i}G'$ is an abelian group. For each element $\gamma_i$ in $\Gamma_i$ (there are $O((\log\abs{G})^2)$ such elements),
the quantum algorithms for Tasks (i) and (ii) in Theorem \ref{tasks} are used to check whether $\abs{\gen{\Gamma_i}G'}=\abs{\gen{\gamma_i,G'}}$. Since 
necessarily $\gen{\gamma_i,G'} \le\gen{\Gamma_i}G'$, this test is sufficient to check whether $\gen{\Gamma_i}G'=\gen{\gamma_i,G'}$.
\item\vspace{-2mm}
The tests at Steps 9, 10 to 17 are done by noticing that $[\Gamma_i,G']=\triv$ if and only if $[\gamma,g'_j]=e$ for each $\gamma\in \Gamma_i$ and each $j\bound{1}{t}$ .
\item\vspace{-2mm}
Testing whether $zwz^{-1}$ is in $\gen{w}$ at Step 23 is done by trying to decompose  $zwz^{-1}$ over $\gen{w}$ using
the quantum algorithm for Task (iii) in Theorem \ref{tasks}, and then checking if the decomposition indeed represents $zwz^{-1}$ (since, a priori,
this algorithm can have an arbitrary behavior when $zwz^{-1}\notin\gen{w}$).
\end{itemize}
This description, along with Theorem \ref{tasks} and with the observation that the sets $U$, $V$ and $\Sigma$ have size $O((\log \abs{G})^2)$, show that all the steps of Procedure  $\proc{Decompose}$ can be implemented in time polynomial in $\log\abs{G}$.
The following theorem states the time complexity of  Procedure $\proc{Decompose}$, and also its correctness.
\begin{theorem}\label{th:correctness}
Let $G$ be a group in the class $\Sp$, given as a black-box group (with unique encoding). 
The procedure $\proc{Decompose}$ on input $G$ outputs, with high probability, a pair $(U,v)$
such that $(\gen{U},\gen{v})$ is a standard decomposition of $G$. 
It can be implemented in time polynomial in $\log\abs{G}$ on a quantum computer.
\end{theorem}
Before giving a complete proof of Theorem \ref{th:correctness} in Subsection \ref{sub_proof}, we first describe its outline below,
which we believe is also instructive in that it describes what procedure $\proc{Decompose}$ actually does.

Suppose that $(A,\gen{y})$ is a standard decomposition of $G$ with $\abs{y}=m$. This decomposition is unknown, and the value of $m$ too. 
Suppose that $\kappa=p_1^{e_1}\cdots p_r^{e_r}$ where the $p_i$'s are distinct prime numbers.
The first thing that is done is to convert the set of generators of $G$ into a set $\Gamma=\cup_{i=1}^r \Gamma_i$ of generators of prime powers
(where each $\Gamma_i$ consists of elements of order $p_i^{k_i}$ with $0\le k_i\le e_i$). 

The idea of the procedure is then to construct two sets: a set $U$ which will contain generators of $A$ and a set $V$ which will contain elements of prime power order 
of the form $ay^\alpha$ with $a\in A$ and $\alpha\not\equiv 0\bmod m$. 
More precisely, most elements of $\Gamma$ can be assigned to either $U$ or $V$ using simple rules (from the properties of groups in the class $\Sp$): 
If the order of an element $g$ of $\Gamma$ is not coprime with $\abs{G'}$, then $g$ should be put in $U$ (Step 9);  
If at least two elements of $\Gamma$ are in the same subset $\Gamma_i$ but do not define a cyclic subgroup (up to elements in the commutator subgroup),
 then they both should be put in $U$ (Step 14);
If an element $g$ of $\Gamma$ does not commute with all the elements of $G'$, then $g$ should be put in $V$ (Step 18; for technical reasons, only one element satisfying this condition from each $\Gamma_i$ 
is put in $V$).

It remains to deal with 
the set $\Sigma$ of elements satisfying neither of these three conditions. For elements $w\in \Sigma$ not commuting with at least one element $z$ in $\Sigma$,
deciding whether $w$ should be put in $U$ or in $V$ can be done by checking whether $zwz^{-1}\in\gen{w}$ or not (Steps 22 and 23). 
The last part of the procedure (Steps 25 to 28) deals with the elements in $\Sigma$ commuting with all elements in $\Sigma$; 
these elements are put as far as possible in $U$ to make $\gen{U}$ as large as possible.

Finally, at Step 29, the product
of all the elements in $V$ is raised to some well chosen power in order to obtain
an element $v$ such that $\gen{v}\cap \gen{U}=\triv$. 
It can be shown that $(\gen{U},\gen{v})$ is then a standard decomposition of $G$.

\subsection{Proof of Theorem \ref{th:correctness}}\label{sub_proof}
We start with two lemmas.
\begin{lemma}\label{lemma:primepower}
Let $G$ be a group in the class $\Sp$
and suppose that $(A,\gen{y})\in\Dd_G^m$. 
Let $w=ay^\alpha$ be an element of $G$ with $a\in A$ and $\alpha\not\equiv 0 \bmod m$.
If the order of $w$ is a prime power, then $a\in G'$.
\end{lemma}
\begin{proof}
If the order of $w$ is a prime power, then it 
is necessarily a prime power $p^r$ dividing $m$ since $\alpha\not\equiv 0 \bmod m$.
Now $e=(ay^\alpha)^{p^r}=xa^{p^r}y^{\alpha p^r}=xa^{p^r}$ where $x$ is some element in $G'$. Thus  $a^{p^r}\in G'\subseteq A$. Since $p^r$ is coprime with $\abs{A}$, 
we conclude that $a\in G'$.
\end{proof}

\begin{lemma}\label{lemma:pair}
Let $G$ be a group in the class $\Sp$ and
suppose that $(A,\gen{y})\in\Dd_G^m$. Let $\Sigma$ be a set of elements of $G$ of prime power order 
such that 
each element of $\Sigma$ has order coprime with $\abs{G'}$ and commutes with all the elements in $G'$.
Let $w$ and $z$ be two elements of $\Sigma$ such that $[w,z]\neq e$. Then
\begin{itemize}
\item[(1)]
if $zwz^{-1}\in\gen{w}$ then $w\in A$ and $z=ay^\alpha$ with $a\in A$ and $\alpha\not\equiv 0\bmod m$.
\item[(2)]
if $zwz^{-1}\notin\gen{w}$ then $z\in A$ and $w=ay^\alpha$ with $a\in A$ and $\alpha\not\equiv 0\bmod m$.
\end{itemize}
\end{lemma}
\begin{proof}
Since $[w,z]\neq e$,  at least one of $w$ and $z$ is of the form $ay^\alpha$ with $a\in A$ and $\alpha\not\equiv 0\bmod m$.
Lemma \ref{lemma:primepower} shows that exactly  one among $w$ and $z$ is of this form, while the other is in $A$
(remember that the elements $w$ and $z$ commute with all the elements in $G'$). 

Let us first prove assertion (1). Suppose that $z\in A$ (and thus necessarily $w=ay^\alpha$ with $a\in A$ and $\alpha\not\equiv 0\bmod m$). 
If $zwz^{-1}\in \gen{w}$, then $[z,w]=(zwz^{-1})w^{-1}$ is in $\gen{w}$ too. 
Since the order of $w$ is necessarily coprime with $|G'|$ (remember that $\abs{w}$ is a prime power and thus divides $m$), 
we conclude that $[z,w]=e$. This gives a contradiction. Thus, if $zwz^{-1}\in \gen{w}$, then $w\in A$.

We now prove assertion (2).
Suppose that $w\in A$ (and thus necessarily $z=ay^\alpha$ with $a\in A$ and $\alpha\not\equiv 0\bmod m$).
Then $zwz^{-1}$ is also in $A$. More precisely, $zwz^{-1}=[z,w]w$.
From the the observation that $zwz^{-1}$ has the same order as $w$ and the fact that $gcd(\abs{w},\abs{G'})=1$, 
we conclude that $[z,w]=e$ and that $zwz^{-1}\in\gen{w}$. Thus, if $zwz^{-1}\not\in \gen{w}$, then $z\in A$.
\end{proof}

We now proceed with the proof of  Theorem \ref{th:correctness}.

\begin{proof}[Proof of Theorem \ref{th:correctness}]
The complexity of Procedure $\proc{Decompose}$ follows from the description of the procedure given in Subsection \ref{sub_descr}. 
It remains to prove its correctness.

Let $(A,\gen{y})$ be a standard decomposition of $G$ with $\abs{y}=m$.
Notice that each call to the quantum algorithms solving the tasks mentioned in Theorem \ref{tasks} realized in the Procedure $\proc{Decompose}$ has success 
probability at least $1-1/\poly(\abs{G})$.
Then, with high probability, there is no failure at those steps. In the following we suppose that this is the case and show that, then, the procedure necessarily outputs a standard decomposition of
$G$.  

First, notice that the sets $\Gamma_i$ constructed in the loop of Steps 5 to 19 are such that $G=\gen{\cup_{i=1}^r\Gamma_i}$.
Moreover, they satisfy the following property:
If $p_i$ divides $m$, then $\gen{\Gamma_i}G'=\gen{y^{m/p_i^{e_i}},G'}$ from Lemma \ref{lemma:primepower}; If $p_i$ does not divides $m$, then
$\gen{\Gamma_i}G'=A_{p_i}G'$, where $A_{p_i}$ denotes the Sylow $p_i$-subgroup of $A$ (since, in this case,
the $\abs{G}/{p_i^{e_i}}$-th power of an element $ay^\alpha$ of $G$ is $xa^{\abs{G}/{p_i^{e_i}}}$ where $x$ is an element of $G'$). 

At the end of the loop of Steps 5 to 19, the set $U\cup V\cup \Sigma$ is a generating set of $G$ (here the fact that $G'\subseteq \gen{U}$ is important).
More precisely, the set $U$ contains only elements of $A$. The set $V$ contains only elements of the 
form $ay^{\alpha_i m/p_i^{e_i}}$ for some $i\bound{1}{r}$ such that $p_i$ divides $m$, 
where $a\in G'$ (from Lemma \ref{lemma:primepower}) and $\alpha_i$ is an integer such that $gcd(\alpha_i,p_i)=1$.
Moreover there is at most one element of this form in $V$ for each $i\bound{1}{r}$ such that $p_i$ divides $m$. 
The set $\Sigma$ is a set of elements 
satisfying the conditions of Lemma \ref{lemma:pair}.

In the loop of Steps 20 to 24, all the elements $w\in \Sigma$ such that $[w,\Sigma]\neq \triv$ are put in either $U$ or $V$.
From Lemma \ref{lemma:primepower} and Lemma \ref{lemma:pair}, the elements put in $U$ are elements of $A$ and the elements put in $V$ are of the form 
$w=ay^{\alpha}$ for some $a\in G'$ and some $\alpha\not\equiv 0\bmod m$.
At the end of the loop, the elements of $\Sigma\backslash (U\cup V)$ are commuting with all the elements of $\Sigma$.

Finally,  the loop of Steps 25 to 28 ensures that all the elements of $\Sigma\backslash(U\cup V)$ are put in either
$U$ or $V$ in the following way. The new elements put in $U$ are precisely those commuting with the original set $U$ (since these new elements also commute together,
the final subgroup $\gen{U}$ will then be abelian).
The elements put in $V$ are such that, at the end of the loop, $V$ contains again only elements of the 
form $ay^{\alpha_i m/p_i^{e_i}}$ with $a\in G'$ and $gcd(\alpha_i,p_i)=1$ for some $i\bound{1}{r}$ such that $p_i$ divides $m$. 
Moreover there is at most one element of this form in $V$ for each such $i$ (from the construction of the set $\Sigma$). This latter observation implies that
the element $z$ constructed at Step 29 is  such that $\gen{z}G'=\gen{V}G'$.

The final subgroup $\gen{U}$ is abelian and, since $G'\subseteq \gen{U}$,
is normal in $G$.
Since $\gen{z}G'=\gen{V}G'$,
we know that $\gen{z,U}=G$ (remember that $G'\subseteq \gen{U}$). 
The element $v$ constructed at Step 29 is 
of the form $ay^\alpha$, with $a\in G'$ and $\alpha$ coprime with $m$, and then 
$\gen{v,U}=G$, but $v$ satisfies the additional relation $v^b=e$.
Since $\gen{U}$ is abelian and each element of $U$ has order coprime with $\abs{v}$,
we conclude that $gcd(\abs{v},\abs{\gen{U}})=1$.
Thus $\gen{v}\cap\gen{U}=\triv$. 

This shows that the output $(U,v)$ of Procedure $\proc{Decompose}$ is such that that 
$(\gen{U},\gen{v})\in\Dd^{m'}_G$ where $m'=\abs{v}\le m$ (more precisely, $\abs{v}$ divides $m$ by construction).  
Since $m$ is the minimal integer such that $\Dd_G^m\neq\varnothing$ 
(because $(A,\gen{y})$ is a standard decomposition of $G$), 
we conclude that $m=m'$ and that Procedure $\proc{Decompose}$ finds a standard decomposition of the group $G$.
\end{proof}

\section{Set Discrete Logarithm}\label{section:set}

\subsection{Statement of the problem}
We first introduce the following useful notation.
Let $\field$ be a finite field, and $\Sigma=\{x_1,\ldots,x_t\}$ be any subset of $\field$ with possible repetitions, i.e.,~all the $x_i$'s are elements of $\field$, but may not be distinct.
For any integer $k$, we denote by $\Sigma^k$ the subset of $\field$ with possible repetitions $\{x_1^k,\ldots,x_t^k\}$.  

In this section we consider the following problem. Here $u$ is a positive integer which is a parameter of the problem 
(taking $u\ge 2$ does not make the problem significantly harder, but this enables us to give a more convenient presentation of 
our results). 

\begin{codebox}
$\proc{Set Discrete Logarithm}$\\\vspace{-5mm}
\zi\const{input:} two lists  $(S_1,\ldots,S_u)$ and $(T_1,\ldots,T_u)$ where, for each integer $h\bound{1}{u}$, 
$S_h$ and $T_h$  
\zi \hspace{20mm} 
are subsets with possible repetitions of some finite field $\field_h$.
\zi\const{output:} a positive integer $k$ such that $T_h^k= S_h$ 
for all $h\bound{1}{u}$, if such an integer exists.\vspace{0mm}
\end{codebox}
Notice that the case $u=1$ with $\abs{S_1}=\abs{T_1}=1$ is the usual discrete logarithm problem over the multiplicative group of the field $\field_1$.
Actually, our algorithm solving the problem  $\proc{Set Discrete Logarithm}$ will only need the multiplicative structure of the fields,
and then also works if we replace  in the definition each field $\field_h$ by any multiplicative finite group $G_h$. 
However, since the main applications of our algorithm deal with field structures (as described in Section \ref{section:conjugacy} and Section \ref{section_algorithm}), 
we  describe our results in the present slightly less general form. 

Given an instance of $\proc{Set Discrete Logarithm}$, let $m_{S}$ denote the smallest positive integer such that $x^{m_{S}}=1$ for all $x\in S_1\cup\cdots\cup S_u$,
and let $m_{T}$ denote the smallest positive integer such that $y^{m_{T}}=1$ for all $y\in T_1\cup\cdots\cup T_u$. 
The main result of this section is the following theorem.

\begin{theorem}\label{theorem_set}
There exists a quantum algorithm that solves with high probability the problem $\proc{Set Discrete}$ $\proc{Logarithm}$, and runs in time polynomial 
in $u$, $\log (m_S+m_T)$, and $\max_{1\le h\le u}(\abs{S_h}+\abs{T_h}+\log\abs{\field_h})$.
\end{theorem}

\subsection{Proof of Theorem \ref{theorem_set}}
We first describe how to compute intersections of cosets of abelian groups efficiently using a quantum computer.
\begin{proposition}\label{coset_int}
Let $\Gamma$ be an abelian group, given as a black-box, and $\Gamma_1,\Gamma_2$ be two subgroups of $\Gamma$ given by generating sets. 
Let $x$ and $y$ be two elements of $\Gamma$.
There exists a quantum algorithm that decides with high probability, in time polynomial in $\log\abs{\Gamma}$, whether $x\Gamma_1\cap y\Gamma_2$ is empty.
Moreover, when the algorithm decides that $x\Gamma_1\cap y\Gamma_2\neq \varnothing$,
it also outputs an element $\gamma\in\Gamma$, and $t=O(\log\abs{\Gamma})$ elements $\gamma_1,\ldots,\gamma_t$
such that  $x\Gamma_1\cap y\Gamma_2=\gamma\gen{\gamma_1,\ldots,\gamma_t}$ with high probability. 
\end{proposition}
\begin{proof}
A standard result of group theory states that the set $x\Gamma_1\cap y\Gamma_2$ is either empty, or is a coset of the subgroup $\Gamma_1\cap \Gamma_2$
(note that this statement is true even if $\Gamma$ is not abelian).
Notice that $x\Gamma_1\cap y\Gamma_2\neq \varnothing$ if and only if $xy^{-1}\in \Gamma_1\Gamma_2$.
This can be checked efficiently using the quantum algorithm by Ivanyos et al.~\cite{Ivanyos+03} testing membership in abelian groups, but more
work is needed to find an explicit element in $x\Gamma_1\cap y\Gamma_2$.

Let $\{\alpha_1,\ldots,\alpha_s\}$ and $\{\beta_1,\ldots,\beta_t\}$ be bases of $\Gamma_1$ and $\Gamma_2$ respectively.
Define the abelian group $P_1=\Int_{\abs{\alpha_1}}\times\cdots\times\Int_{\abs{\alpha_s}}\times\Int_{\abs{\beta_1}}\times\cdots\times\Int_{\abs{\beta_t}}\times\Int_{\abs{xy^{-1}}}$ 
and define the map $f_1$ from $P_1$ to $\Gamma$ as follows: for any $(a_1,\ldots,a_s,b_1,\ldots,b_t,c)$ in $P_1$,
$$
f_1(a_1,\ldots,a_s,b_1,\ldots,b_t,c)=\alpha_1^{a_1}\cdots\alpha_s^{a_s}\beta_1^{b_1}\cdots\beta_t^{b_t}x^{-c}y^{c}.
$$
Notice that the set
$
Q_1=\{(a_1,\ldots,a_s,b_1,\ldots,b_t,c)\in P_1\st x^cy^{-c}=\alpha_1^{a_1}\cdots\alpha_s^{a_s}\beta_1^{b_1}\cdots\beta_t^{b_t}\}
$
is a subgroup of $P_1$, and that the function $f_1$ is constant on cosets of $Q_1$ in $P_1$, with distinct values on distinct cosets. 
This is thus an instance of the abelian HSP, and a set of generators of $Q_1$ can be found in time polynomial in $\log\abs{P_1}=O(\log\abs{\Gamma})$.
The set  $x\Gamma_1\cap y\Gamma_2$ is not empty if and only if $Q_1$ contains some element of the form $(a_1,\ldots,a_s,b_1,\ldots,b_t,1)$, 
in which case the element $\gamma=x\alpha_1^{-a_1}\cdots\alpha_s^{-a_s}$ is in $x\Gamma_1\cap y\Gamma_2$.

We now show how to compute a generating set of the subgroup $\Gamma_1\cap\Gamma_2$. This can be done using the
quantum algorithm by Friedl et al.~\cite{Friedl+STOC03} computing the intersection of subgroups in ``smoothly solvable'' groups, but we present 
here a much simpler quantum algorithm for the abelian case, inspired by techniques developed in \cite{McKenzie+SICOMP87}.
Let $f_2$ be the map from the abelian group $P_2=\Int_{\abs{\alpha_1}}\times\cdots\times\Int_{\abs{\alpha_s}}\times\Int_{\abs{\beta_1}}\times\cdots\times\Int_{\abs{\beta_t}}$ 
to $K_1K_2$ defined as follows:  for any $(a_1,\ldots,a_s,b_1,\ldots,b_t)$ in $P_2$,
$$
f_2(a_1,\ldots,a_s,b_1,\ldots,b_t)=\alpha_1^{a_1}\cdots\alpha_s^{a_s}\beta_1^{b_1}\cdots\beta_t^{b_t}.
$$
Notice that the set
$
Q_2=\{(a_1,\ldots,a_s,b_1,\ldots,b_t)\in P_2\st \alpha_1^{a_1}\cdots\alpha_s^{a_s}\beta_1^{b_1}\cdots\beta_t^{b_t}=1\}
$
is a subgroup of $P_2$, and that the function $f_2$ is constant on cosets of $Q_2$ in $P_2$, with distinct values on distinct cosets. 
This is thus an instance of the abelian HSP, and a set of generators $\{z_1,\ldots,z_r\}$ of $Q_2$ with $r=\log\abs{\Gamma}$  
can be found in time polynomial in $\log\abs{P_2}=O(\log\abs{\Gamma})$ using the algorithm described in Subsection \ref{sub:clBB} .
For each $i\bound{1}{r}$ let us write $z_i=(u_{i1},\ldots,u_{is},v_{i1},\ldots,v_{it})$ and define $\gamma_i=\alpha_1^{u_{i1}}\cdots\alpha_s^{u_{is}}$. 
Then it is easy to check that $\Gamma_1\cap\Gamma_2=\gen{\gamma_1,\ldots,\gamma_r}$. 
\end{proof}

We are now ready to give our proof of Theorem \ref{theorem_set}.
\begin{proof}[Proof of Theorem \ref{theorem_set}]
For the sake of brevity, let us denote $\Sigma=S_1\cup\cdots\cup S_u\cup T_1\cup\cdots\cup T_u$.
We first compute the orders of all the elements in $\Sigma$ using Shor's algorithm \cite{ShorSICOMP97}.
The value $m_S$ is the least common multiple of the orders of all the elements in $S_1\cup\cdots\cup S_u$, and 
the value $m_T$ is the least common multiple of the orders of all the elements in $T_1\cup\cdots\cup T_u$.
The values $m_S$ and $m_T$ can then be computed in time polynomial in $\log (m_S+m_T)$, $\abs{\Sigma}$, and $\max_{1\le h\le u}\log\abs{\field_h}$.
Notice that, for any positive integer $k$, the least common multiple of the orders of all the elements in $T_1^k\cup\cdots\cup T_u^k$ is $m_T/gcd(k,m_T)$.
Then, if $m_S$ does not divide $m_T$, then there is no solution to the problem $\proc{Set Discrete Logarithm}$. 
If $m_S$ divides $m_T$ but $m_S\neq m_T$, then a solution (if it exists) can be found by replacing the list $(T_{1},\ldots,T_u)$ by the list $(T_1^{m_T/m_S},\ldots,T_u^{m_T/m_S})$. 
Thus, without loss of generality, we 
suppose hereafter that $m_S=m_T$ and denote by $m$ this value. Then a solution $k$ can be searched for in the set $\Int_m^\ast$.

Let $\{m_{1},\ldots,m_{\ell}\}=\cup_{z\in\Sigma}\{\abs{z}\}$ denote the set of orders of the elements in $\Sigma$.
For each $h\bound{1}{u}$ and each $i\bound{1}{\ell}$, we define the subsets
$$
S_{h,i}=\{x\in S_h\:\vert\:\abs{x}=m_{i}\} \textrm{ and } T_{h,i}=\{y\in T_h\:\vert\:\abs{y}=m_{i}\}. 
$$
Let us also define the sets 
$$
K_{h,i}=\{k\in\Int_m^\ast\st T_{h,i}^k=S_{h,i}\} \textrm{ and }
\overline{K}_{h,i}=\{k\in\Int_m^\ast\st T_{h,i}^k=T_{h,i}\}.
$$
It is straightforward to check that the set $\overline{K}_{h,i}$ is a subgroup of $\Int_m^\ast$, 
and that  the set  $K_{h,i}$ is either empty, or is a coset of  $\overline{K}_{h,i}$  in $\Int_m^\ast$. 

Let $K\subseteq \Int_m^\ast$ denote the set of solutions of the instance of $\proc{Set Discrete Logarithm}$ we are considering. Then
$$
K=\bigcap_{1\le h\le u}\Big(\bigcap_{1\le i\le \ell}K_{h,i}\Big).
$$
The set $K$ can be computed efficiently by applying successively 
the quantum algorithm of Proposition \ref{coset_int}
if, for each $h\bound{1}{u}$ and each $i\bound{1}{\ell}$, the set $K_{h,i}$ is known 
(more precisely, if a generating set of $\overline{K}_{h,i}$ and an element of $K_{h,i}$ are known).

The final part of the proof shows how to compute these sets $K_{h,i}$.
Let us fix an integer $h\bound{1}{u}$ and an integer $i\bound{1}{\ell}$. We suppose that $S_{h,i}$ and $T_{h,i}$ have the same size (otherwise $K_{h,i}=\varnothing$ and thus $K=\varnothing$).
Denote $S_{h,i}=\{x_1,\ldots,x_{v}\}$ and $T_{h,i}=\{y_1,\ldots,y_{v}\}$, where $v=\abs{S_{h,i}}$ 
depends on $h$ and $i$. We present a quantum procedure computing
a set of generators of $\overline{K}_{h,i}$, and an element
$k_{h,i}$ in $K_{h,i}$ when this set is not empty, in time polynomial in $v$, $\log m$, and $\log\abs{\field_h}$.

We first show how to compute the subgroup $\overline{K}_{h,i}$.
Let $\prec$ be an arbitrary strict total ordering of the elements of 
$\field_h$. 
Without loss of generality we can suppose that $x_1\preceq x_2\preceq\cdots\preceq x_{v}$. 
Let $\mu$ be the function from $\Int_m^\ast\times\{1,\ldots,v\}$ to $\field_h$ defined as follows: 
for any $k\in\Int_m^\ast$ and any $j\bound{1}{v}$, $\mu(k,j)$ is the $j$-th element (with respect to the order $\prec$)
of the set $T_{h,i}^k$.
Let $f$ be the function from $\Int_m^\ast$ to $(\field_h)^{v}$ such that, for any $k\in\Int_m^\ast$:
$$f(k)=(\mu(k,1)y_1^{-1},\ldots,\mu(k,v)y_{v}^{-1}).$$ 
Notice that the set $\{k\in\Int_m^\ast\st f(k)=(1,\ldots,1)\}$ is precisely the subgroup $\overline{K}_{h,i}$ of $\Int_m^\ast$.
Moreover, the function $f$ is constant on cosets of $\overline{K}_{h,i}$ in $\Int_m^\ast$, with distinct values on distinct cosets
(since $f(k_1)=f(k_2)$ implies that $T_{h,i}^{k_1}=T_{h,i}^{k_2}$ and thus $k_1\in k_2 \overline{K}_{h,i}$).
This is thus an instance of the abelian HSP, and a set of generators of $\overline{K}_{h,i}$ can be found in time polynomial in
$v$, $\log m$ and $\log\abs{\field_h}$ using the algorithm described in Subsection \ref{sub:clBB} 
(notice that the underlying group is $\Int_m^\ast$, and that the value of the function $f$ can be computed in time $v$, $\log m$ and $\log\abs{\field_h}$). 

We now show how to compute an element $k_{h,i}$ in $K_{h,i}$ if this set is not empty. 
We first try to find an element $\alpha\in\Int_{m_i}^\ast$ such that $T_{h,i}^{\alpha}=S_{h,i}$.
This is done by, for each $j\bound{1}{v}$, trying to find an integer $\alpha_j\in\Int_{m_{i}}^\ast$ such that
$x_{1}^{\alpha_j}=y_{j}$, if such an integer exists (notice that, for each $j$, there is at most one element $\alpha_j$ in $\Int_{m_{i}}^\ast$ satisfying this condition,
which can be computed in time polynomial in $\log m_i$ and $\log\abs{\field_h}$ using the quantum algorithm for the standard discrete logarithm problem \cite{ShorSICOMP97}) and checking whether
$T_{h,i}^{\alpha_j}=S_{h,i}$.  
If no such value $\alpha$ can be found, we conclude that $K_{h,i}$ is empty.
Otherwise we take any such value $\alpha$ and compute $k_{h,i}$ as follows.
Let us write the prime power decomposition of $m$ as 
$m=p_1^{\epsilon_1}\cdots p_r^{\epsilon_r}{p'}_1^{\eta_1}\cdots {p'}_{s}^{\eta_s}q_1^{\delta_1}\cdots q_t^{\delta_t}$,
where each prime $p_l$ divides $m_{i}$ for $l\bound{1}{r}$, each prime $p'_l$ divides $\alpha$ but not $m_{i}$ for $l\bound{1}{s}$, 
and each prime $q_l$ divides neither $m_{i}$ nor $\alpha$ for $l\bound{1}{t}$.
Then the integer
$$
k_{h,i}=\alpha+m_{i}q_1^{\delta_1}\cdots q_t^{\delta_t} \bmod m
$$
is coprime with $m$ (since $\alpha$ is coprime with $m_i$ and then each prime $p_l$, $p'_l$ or $q_l$ does not divide $k_{h,i}$), 
and hence is in $\Int_m^\ast$. From the choice of $\alpha$ and since any element in $T_{h,i}$ has order $m_{i}$, we conclude that 
$k_{h,i}$ is in the set~$K_{h,i}$.
\end{proof}

\section{Discrete Logarithm up to Conjugacy}\label{section:conjugacy}

\subsection{Statement of the problem}
Given a positive integer $r$ and a finite field $\field$, remember that $GL(r,\field)$ denotes the multiplicative group of invertible matrices of
size $r\times r$ with entries in $\field$. In this section we consider the following problem. Here $u$ is again a positive integer which is a parameter of the problem.

\begin{codebox}
$\proc{Discrete Log up to Conjugacy}$\\\vspace{-5mm}
\zi\const{input:} two lists of matrices  $(M^{(1)}_1,\ldots,M^{(u)}_1)$ and $(M^{(1)}_2,\ldots,M^{(u)}_2)$ where, for each integer  
\zi \hspace{12mm} $h\bound{1}{u}$, $M_1^{(h)}$ and $M_2^{(h)}$ are in $GL(r_h,\field_h)$ for some positive integer $r_h$ and some
\zi \hspace{12mm}  finite field $\field_h$.
\zi\const{output:} a positive integer $k$  and $u$ matrices $M^{(h)}\in GL(r_h,\field_h)$ such that 
\zi \hspace{15mm}  $M^{(h)}\cdot M^{(h)}_1= [M^{(h)}_2]^k\cdot M^{(h)}$  for each $h\bound{1}{u}$, if such elements exist.\vspace{0mm}
\end{codebox}
In the statement of the above problem, the notation $[M^{(h)}_2]^k$ simply means $M^{(h)}_2$ raised to the $k$-th power.
Notice that the case $u=1$ and $r_1=1$ is basically the usual discrete logarithm problem over the multiplicative group of the finite field $\field_1$.

Let $m_1$ and $m_2$ denote the smallest positive integers such that $[M_1^{(h)}]^{m_1}=I$ and $[M_2^{(h)}]^{m_2}=I$ for all $h\bound{1}{u}$.
The main result of this section is the following theorem. 
 \begin{theorem}\label{theorem:conjugacy}
 There exists a quantum algorithm that solves with high probability the problem $\proc{Discrete Log}$ $\proc{up to Conjugacy}$, and runs in time polynomial in $u$, $\log (m_1+m_2)$, and  $\max_{1\le h\le u}(r_h+\log\abs{\field_h})$
 \end{theorem}
 
 \subsection{Proof of Theorem \ref{theorem:conjugacy}}
 The quantum algorithm solving the problem $\proc{Discrete Log up to}$ $\proc{Conjugacy}$ follows from a reduction to the problem $\proc{Set Discrete Logarithm}$.
 The key idea is to represent each matrix by its set of elementary divisors.  
  We will first introduce some definitions 
 and prove two lemmas before moving to the proof of Theorem \ref{theorem:conjugacy}. 
 In this subsection we use the notations introduced in Subsection \ref{sub:inv}.

Let $M$ be a matrix in $GL(r,\field)$, where $r$ is a positive integer and $\field$ is a finite field.
The minimal polynomial $a_s(x)$ of $M$ has not in general all its roots in $\field$, and, 
in order to define the elementary divisors of $M$, we need then to work
on a field extension of $\field$ containing all the roots of $a_s(x)$. 
Denote $\field=GF(q)$ where $q$ is some prime power. It is well known that the roots of any irreducible factor of degree $d$ of a polynomial in $\field[x]$ are elements of the
field extension $GF(q^d)$ of $\field$ (see \cite{Lidl+08} for example).
Then the field extension $GF(q^{d'})$ splits the polynomial $a_s(x)$, where $d'$ denotes the least common multiple
of the degrees of the irreducible factors of $a_s(x)$ over $\field$. 
However, the value $d'$ can be in general superpolynomial in $r$, and thus we need to be more careful to obtain an algorithm with running time 
polynomial in $r$ and $\log \abs{\field}$. This is why we introduce the following 
definition (we also take in consideration the degrees of the associated elementary divisors for technical reasons).

\begin{definition}\label{definition:sets}
Let $M$ be a matrix in $GL(r,\field)$ where $r$ is a positive integer and $\field=GF(q)$ is a finite field of prime power order $q$, 
and let $d$ and $\ell$ be two positive integers.
Suppose that $\{(x-\lambda_1)^\ell,\ldots,(x-\lambda_t)^\ell\}$ is the subset of all elementary divisors of degree $\ell$ of $M$  such that 
each $\lambda_i$ is an element in $GF(q^d)$ but is not in any proper subfield of $GF(q^d)$.
Then we define $\Sigma_{d,\ell}(M)$ as the subset of $GF(q^d)$ with possible repetitions $\{\lambda_1,\ldots,\lambda_t\}$.
\end{definition}
\noindent{\bf Example}. Define the two polynomials $f_1=(x^2+x+1)$ and $f_2=(x^2+x+1)^2(x^3+x+1)$ over $GF(2)$, and the matrix $M=diag(C_1,C_1,C_2)$ 
where $C_1$ (resp.~$C_2$) denotes the companion matrix of $f_1$ (resp.~$f_2$). Notice that $x^2+x+1$ and $x^3+x+1$ are irreducible over $GF(2)$.
The matrix $M$ has size $11\times 11$, consists of 3 diagonal 
blocks of size $2\times 2$, $2\times 2$ and $7\times 7$ respectively, and is actually already in rational normal form.
In particular, its invariant factors are $(f_1,f_1,f_2)$. 
Then the minimal polynomial of $M$ is $f_2$, which is split by $GF(2^6)$. 
It can be checked that there exist two elements $\alpha_2\in GF(2^2)$ and $\alpha_3\in GF(2^3)$ 
of multiplicative order respectively 3 and 7 such that the polynomial $(x^2+x+1)$ factorizes into $(x-\alpha_2)(x-\alpha_2^2)$
over  $GF(2^2)$ and the polynomial $(x^3+x+1)$ factorizes into $(x-\alpha_3)(x-\alpha_3^2)(x-\alpha_3^4)$ over $GF(2^3)$. 
Then the set of elementary  divisors of $M$ is 
$\{(x-\alpha_2),(x-\alpha_2),(x-\alpha^2_2),(x-\alpha^2_2),(x-\alpha_2)^2,(x-\alpha^2_2)^2,(x-\alpha_3),(x-\alpha^2_3),(x-\alpha^4_3)\}$
and the only sets $\Sigma_{d,\ell}(M)$ that are not empty are 
$\Sigma_{2,1}=\{\alpha_2,\alpha_2,\alpha_2^2,\alpha_2^2\}$,
$\Sigma_{2,2}=\{\alpha_2,\alpha_2^2\}$ and
$\Sigma_{3,1}=\{\alpha_3,\alpha_3^2,\alpha_3^4\}$.\qed\vspace{2mm}

We will need the following result on Jordan matrices.
 \begin{lemma}\label{lemma:jordanpower}
Let $\lambda$ be a nonzero element in a finite field $\fieldK$ and $c$ be a positive integer. 
Let $k$ be a positive integer coprime with the multiplicative order of $J(\lambda,c)$. 
Then the set of elementary divisors of the matrix $[J(\lambda,c)]^k$ is 
$\{(x-\lambda^k)^c\}$.
\end{lemma}
\begin{proof}
Let us write $M=J(\lambda,c)$ and denote by $p$ the characteristic of $\fieldK$. 
The result is trivial if $c=1$ so we suppose that $c\ge 2$.

Our proof is based on the simple fact that the $k$-th power of M is an upper triangular matrix
with $\lambda^k$ along the main diagonal, $k\lambda^{k-1}$ along the first superdiagonal, and possibly other nonzero entries
in the other superdiagonals if $c>2$ (the values of these entries are easy to calculate, but not relevant to this proof).
Let $m$ denote the multiplicative order of $M$. Then, since $M^m=I$ and $\lambda\neq 0$, we have $m\lambda^{m-1}=0$.
Then $p$ divides $m$.

Let $k$ be a positive integer coprime with $m$. Then $k$ is necessary coprime with $p$ from the above observation. 
Notice that  a matrix in $GL(c,\fieldK)$ has $\{(x-\lambda^k)^c\}$ as set of elementary divisors if and only if $(x-\lambda^k)^c$ is its minimal polynomial.
Since the characteristic polynomial of $M^k$ is $(x-\lambda^k)^c$, the minimal polynomial of $M^k$  divides $(x-\lambda^k)^c$.
We now show that $(M^k-\lambda^kI)^{c-1}\neq 0$.
From the description of $M^k$ given above, it is easy to show that $(M^k-\lambda^kI)^{c-1}$ is the matrix where the only nonzero entry 
is located at the first row and the $c$-th column. The value of this entry is $(k\lambda^{k-1})^{c-1}$. 
Since $k$ is coprime with $p$ and $\lambda\neq 0$, we conclude that $(M^k-\lambda^kI)^{c-1}\neq 0$. 
\end{proof}
Since two matrices are similar if and only if they have the same elementary divisors,
Lemma \ref{lemma:jordanpower} shows that a Jordan matrix raised to a power coprime with 
its order is similar to itself.
We now prove the following lemma (remember that, if $\Sigma=\{x_1,\ldots,x_t\}$ is a subset of $\field$ with possible repetitions, 
we denote by $\Sigma^k$ the subset of $\field$ with possible repetitions $\{x_1^k,\ldots,x_t^k\}$).
\begin{lemma}\label{claim1}
Let $M_1$ and $M_2$ be two matrices in $GL(r,\field)$, where $r$ denotes a positive integer and $\field$ denotes a finite field.
Let $m$ be an integer such that $M_1^m=M_2^m=I$, and $k$  be an integer in $\Int^\ast_{m}$. Then $M_1$ and $M_2^k$ are similar in 
$GL(r,\field)$ if and only if, for all positive integers $d$ and $\ell$, the equality
$[\Sigma_{d,\ell}(M_2)]^k=\Sigma_{d,\ell}(M_1)$ holds.
\end{lemma}
\begin{proof}
Let $\fieldK$ be a field extension of $\field$ splitting the minimal polynomial of $M_2$. 
Denote by $(x-\mu_1)^{v_1},\ldots,(x-\mu_s)^{v_s}$ 
the elementary divisors of $M_2$ (where the $\mu_i$'s are elements of $\fieldK$ that may not be distinct).
If $k$ is coprime with $m$, then Lemma \ref{lemma:jordanpower} implies (using the concept of the Jordan normal form) that the elementary divisors of $M_2^k$ are 
$(x-\mu_1^k)^{v_1},\ldots,(x-\mu_s^k)^{v_s}$.
Since two matrices are similar in $GL(r,\field)$ if and only if they have the same elementary divisors, the claim follows from the fact that, if $\fieldK_i$ is the smallest subfield of $\fieldK$ containing $\mu_i$, then $\fieldK_i$ is also the smallest
subfield of $\fieldK$ containing $\mu_i^k$ (since $k$ is coprime with the order of $\mu_i$).
\end{proof}

 We now present the proof of Theorem \ref{theorem:conjugacy}.

 \begin{proof}[Proof of Theorem \ref{theorem:conjugacy}]
Remember that $m_1$ and $m_2$ denote the minimal positive integers such that $[M_1^{(h)}]^{m_1}=I$ and $[M_2^{(h)}]^{m_2}=I$ for all $h\bound{1}{u}$.
Notice that, if $m_1$ does not divide $m_2$, then there is
no solution to the problem $\proc{Discrete Log up to Conjugacy}$. 
If $m_1$ divides $m_2$ but $m_1\neq m_2$, then a solution (if it exists) can be found by replacing each matrix $M_2^{(h)}$ by $[M_2^{(h)}]^{m_2/m_1}$. 
Thus, without loss of generality, we 
suppose hereafter that $m_1=m_2$ and denote by $m$ this value. Then a solution $k$ can be searched for in the set $\Int_m^\ast$.
 
Let us fix an integer $h\bound{1}{u}$ and suppose that $\field_h=GF(q_h)$, where $q_h$ is a some prime power.
We first compute the invariant factors over $\field_h$ of $M_1^{(h)}$ and $M_2^{(h)}$. This can be done
in $O({r_h}^3)$ field operations, using for example the algorithm by Storjohann \cite{StorjohannFOCS01}.
We then factor over $\field_h$ these invariant factors using the Cantor-Zassenhaus algorithm
\cite{Cantor+81}, running in time polynomial in $r_h$ and $\log\abs{\field_h}$.
Let us denote by $D_h$ the set of degrees of the irreducible factors (over $\field_h$) appearing in at least one of these invariant factors.
Notice that obviously $\abs{D}\le 2 r_h$ since each $M_1^{(h)}$ and $M_2^{(h)}$ has at most $r_h$ invariant factors. 
For each $d\in D_h$ and each integer $\ell\bound{1}{r_h}$, we compute the sets $\Sigma_{d,\ell}(M_1^{(h)})$ 
and $\Sigma_{d,\ell}(M_2^{(h)})$ defined in Definition \ref{definition:sets} as follows: 
the irreducible factors of degree $d$ of the invariant factors of $M_1^{(h)}$ and $M_2^{(h)}$ are factorized over $GF(q^d)$
using the Cantor-Zassenhaus algorithm \cite{Cantor+81}, and the elementary factors of degree $\ell$ are then collected.

Lemma \ref{claim1} implies that there exists a solution to the problem $\proc{Discrete Log up to Conjugacy}$ if and only if there exists some integer $k\in\Int_m^\ast$ such that 
$[\Sigma_{d,\ell}(M_2^{(h)})]^k=\Sigma_{d,\ell}(M_1^{(h)})$ for all integers $h\bound{1}{u}$, all integers $d\in D_h$ and all integers $\ell\bound{1}{r_h}$.
Such an integer $k$ (if it exists) can then be found with high probability using the quantum algorithm of Theorem \ref{theorem_set} in time polynomial in
$u$, $\log m$, and $\max_{1\le h\le u}(r_h+\log\abs{\field_h})$.

Finally, if such a solution $k$ exists, then, for each $h\bound{1}{h}$, a matrix $M^{(h)}\in GL(r_h,\field_h)$
such that $M^{(h)}M^{(h)}_1=[M^{(h)}_2]^kM^{(h)}$ can then be computed for this value of $k$ in time polynomial in $r_h$ and $\log\abs{\field_h}$ 
using efficient classical algorithms, for example the algorithm by Storjohann \cite{StorjohannFOCS01}. 
 \end{proof}

\section{Proof of Theorem \ref{theorem_main}}\label{section_algorithm}

We first state some technical results by Le Gall \cite{LeGallSTACS09} we use to prove Theorem \ref{theorem_main}.
We will first need  the following result from \cite{LeGallSTACS09} that shows necessary and sufficient conditions
for the isomorphism of two groups in the class~$\Sp$. 
\begin{proposition}[Proposition 5.1 in \cite{LeGallSTACS09}]\label{proposition_class}
Let $G$ and $H$ be two groups in $\Sp$.
Let $(A_1,\gen{y_1})$ and $(A_2,\gen{y_2})$ be standard decompositions
of $G$ and $H$ respectively and let $\varphi_1\in Aut(A_1)$ (resp.~$\varphi_2\in Aut(A_2)$) be 
the action by conjugation of $y_1$ on $A_1$ (resp.~of $y_2$ on $A_2$).
The groups $G$ and $H$ are isomorphic if and only if the following three conditions 
hold:
(i)
$A_1\cong A_2$; and 
(ii)
$\abs{y_1}=\abs{y_2}$; and
(iii)
there exists a positive integer $k$ and an isomorphism $\chi\colon A_1\to A_2$
such that $\varphi_1=\chi^{-1}\varphi_2^k\chi$, where $\varphi_2^k$ means $\varphi_2$ composed by itself $k$ times.
\end{proposition}

From now, we identify, for any prime $p$, the finite field of size $p$ with $\Int_p$.
The following proposition summarizes key elements used in the classical algorithm by Le Gall \cite{LeGallSTACS09} that we will need.

\begin{proposition}[\cite{LeGallSTACS09}]\label{theorem_red}
Let $A_1$ and $A_2$ be two isomorphic abelian groups.
Let $(g_1,\ldots,g_s)$ and $(h_1,\ldots,h_s)$ be bases of $A_1$ and $A_2$ respectively. 
Suppose that $A_1\cong (\Int_{p_1^{f_1}})^{r_1}\times \cdots\times(\Int_{p_t^{f_t}})^{r_t}$, where each $r_i$ is a positive integer, and 
each $p_i$ is a prime but $p_i^{f_i}\neq p_j^{f_j}$ for $i\neq j$. Denote $\mathsf{V}=GL(r_1,\Int_{p_1})\times\cdots\times GL(r_t,\Int_{p_t})$.
Then there exists two homomorphisms $\Phi_1\colon Aut(A_1)\to \mathsf{V}$ and $\Phi_2\colon Aut(A_2)\to \mathsf{V}$ such that, for any two automorphisms $\zeta_1\in Aut(A_1)$ and $\zeta_2\in Aut(A_2)$ of order coprime with $\abs{A_1}$,  the following two assertions are equivalent:
\begin{itemize}
\item[(i)]
there exists an  isomorphism $\chi\colon A_1\to A_2$ such that $\zeta_1=\chi^{-1}\zeta_2\chi$;
\item[(ii)] 
there exists an element  $X\in \mathsf{V}$ such that $\Phi_1(\zeta_1)=X^{-1}\Phi_i(\zeta_2)X$. 
\end{itemize}
Moreover, if, for each $j\bound{1}{s}$, integers $u_{ij}$ and $v_{ij}$ such that 
$\zeta_1(g_j)=g_1^{u_{1j}}\cdots g_s^{u_{sj}}$ and $\zeta_2(h_j)=h_1^{v_{1j}}\cdots h_s^{v_{sj}}$ are known,
then the following holds:
\begin{itemize}
\item[(a)]
the images $\Phi_1(\zeta_1)$ and $\Phi_2(\zeta_2)$ can be computed (classically) in time polynomial in $\log\abs{A_1}$;
\item[(b)]
given an explicit element  $X\in \mathsf{V}$ such that $\Phi_1(\zeta_1)=X^{-1}\Phi_i(\zeta_2)X$, 
an isomorphism $\chi:A_1\to A_2$ such that $\zeta_1=\chi^{-1}\zeta_2\chi$ can be computed (classically) in time polynomial in $\log\abs{A_1}$.
\end{itemize}
\end{proposition}

We now present our proof of Theorem \ref{theorem_main}.
\begin{proof}[Proof of Theorem \ref{theorem_main}]
Suppose that $G$ and $H$ are two groups in the class $\Sp$.
In order to test whether these two groups are isomorphic, we first run
Procedure $\proc{Decompose}$ on $G$ and $H$ and obtain outputs $(U_1,y_1)$ 
and $(U_2,y_2)$ such that $(\gen{U_1},\gen{y_1})$ and 
$(\gen{U_2},\gen{y_2})$ are standard decompositions of $G$ and $H$ 
respectively with high probability (from Theorem \ref{th:correctness}).
The running time of this step is polynomial in the logarithms of $\abs{G}$ and $\abs{H}$, from Theorem~\ref{th:correctness}.
Denote $A_1=\gen{U_1}$ and $A_2=\gen{U_2}$. The orders of $A_1,A_2$, $y_1$ and $y_2$ are then computed using the quantum algorithms
for Tasks (i) and (ii) in Theorem \ref{tasks}.
Notice that $\abs{G}=\abs{A_1}\cdot\abs{y_1}$ and $\abs{H}=\abs{A_2}\cdot\abs{y_2}$.  
If $\abs{G}\neq \abs{H}$, we conclude that $G$ and $H$ are not isomorphic. In the following, we suppose that 
$\abs{G}=\abs{H}$ and denote by $n$ this order.

If $\abs{y_1}\neq\abs{y_2}$ we conclude
that $G$ and $H$ are not isomorphic,  from Proposition \ref{proposition_class}. Otherwise denote
 $\abs{y_1}=\abs{y_2}=m$.
Then we compute a basis $(g_1,\ldots,g_{s})$ of $A_1$ and a basis $(h_1,\ldots,h_{s'})$ of $A_2$ 
using the quantum algorithm for Task (ii) in Theorem \ref{tasks}. 
Given these bases it is easy to check the isomorphism of $A_1$ and $A_2$:
the groups $A_1$ and $A_2$ are isomorphic if and only if $s=s'$ and there exists a permutation $\sigma$ of $\{1,\ldots,s\}$ such that $\abs{g_i}=\abs{h_{\sigma(i)}}$
for each $i\bound{1}{s}$.
If $A_1\not\cong A_2$ we conclude that $G$ and $H$ are not 
isomorphic, from Proposition \ref{proposition_class}. 

Now suppose that $A_1\cong A_2\cong (\Int_{p_1^{f_1}})^{r_1}\times \cdots\times(\Int_{p_t^{f_t}})^{r_t}$, where
each $p_i$ is a prime, but $p_i^{f_i}\neq p_j^{f_j}$ for $i\neq j$.
We want to decide whether the action by conjugation $\varphi_1\in Aut(A_1)$ of $y_1$ on $A_1$ 
and the action by conjugation $\varphi_2\in Aut(A_2)$ of $y_2$ on $A_2$ 
satisfy Condition (iii) in Proposition \ref{proposition_class}.
Notice that, for each $j\bound{1}{s}$, we can compute (in time polynomial in $\log n$) integers $u_{ij}$ and $v_{ij}$ such that $\varphi_1(g_j)=y_1g_jy_1^{-1}=g_1^{u_{1j}}\cdots g_s^{u_{sj}}$
and $\varphi_2(h_j)=y_2h_jy_2^{-1}=h_1^{v_{1j}}\cdots h_s^{v_{sj}}$ using the quantum algorithm for Task (iii) in Theorem \ref{tasks}. 
From Proposition \ref{theorem_red}, the images $\Phi_1(\varphi_1)$ and $\Phi_2(\varphi_2)$  can then be computed in time polynomial in $\log n$.
Notice that $[\Phi_1(\varphi_1)]^m=[\Phi_2(\varphi_2)]^m=I$.
 
Since the maps $\Phi_2$ is a homomorphism, Proposition \ref{theorem_red} implies that 
there exists a positive integer $k$ and an isomorphism $\chi:A_1\to A_2$
such that $\varphi_1=\chi^{-1}\varphi_2^k\chi$
if and only if  $\Phi_1(\varphi_1)$ and $[\Phi_2(\varphi_2)]^k$ 
are conjugate in the group $\mathsf{V}=GL(r_1,\Int_{p_1})\times\cdots\times GL(r_t,\Int_{p_t})$.
If we denote $\Phi_1(\varphi_1)=(M_1^{(1)},\ldots,M_1^{(t)})$ and $\Phi_2(\varphi_2)=(M_2^{(1)},\ldots,M_2^{(t)})$, where each $M_1^{(\ell)}$ and each
$M_2^{(\ell)}$ are matrices in $GL(r_\ell,\Int_{p_\ell})$, then
checking if the later condition holds becomes an instance of the problem $\proc{Discrete Log up to Conjugacy}$, 
and can be decided using the algorithm of Theorem \ref{theorem:conjugacy} in time polynomial in $t$, $\log m$, and  $\max_{1\le \ell\le t}(r_\ell+\log p_\ell)$,
i.e., in time polynomial in $\log n$.

If the above instance of  $\proc{Discrete Log up to Conjugacy}$ has no solution, we conclude that $G$ and
$H$ are not isomorphic. Otherwise we take one value $k$ such that each $\Phi_1(\varphi_1)$ and $[\Phi_2(\varphi_2)]^k$ are conjugate, along with an element 
$X\in \mathsf{V}$ such that $X\Phi_1(\varphi_1)=[\Phi_2(\varphi_2)]^kX$ (such an element is obtained from the output of 
the algorithm of Theorem \ref{theorem:conjugacy}), 
and compute an isomorphism $\chi$ from $A_1$ to $A_2$ such that $\varphi_1=\chi^{-1}\varphi_2^k\chi$ using the last part of Proposition \ref{theorem_red}.
The map $\mu: G\to H$ defined as $\mu(xy_1^{j})=\chi(x)y_2^{kj}$ for any $x\in A_1$ and any $j\bound{0}{m-1}$ is 
then an isomorphism from $G$ to $H$ (a detailed proof of this statement can be found in the proof of Proposition \ref{proposition_class} included in~\cite{LeGallSTACS09}).
\end{proof}

\section*{Acknowledgments}
The author is indebted to Yoshifumi Inui for many discussions on similar topics.
He also thanks  Erich Kaltofen, Igor Shparlinski and Yuichi Yoshida for helpful comments.


\begin{thebibliography}{10}

\bibitem{Arvind+CCC04}
{\sc Arvind, V., and Tor{\'a}n, J.}
\newblock Solvable group isomorphism.
\newblock In {\em Proceedings of the 19th IEEE Conference on Computational
  Complexity\/} (2004), pp.~91--103.
\vspace{-1mm}

\bibitem{BabaiSTOC85}
{\sc Babai, L.}
\newblock Trading group theory for randomness.
\newblock In {\em Proceedings of the 17th annual ACM Symposium on Theory of
  Computing\/} (1985), pp.~421--429.
  \vspace{-1mm}

\bibitem{BabaiSTOC91}
{\sc Babai, L.}
\newblock Local expansion of vertex-transitive graphs and random generation in
  finite groups.
\newblock In {\em Proceedings of the 23rd Annual ACM Symposium on Theory of
  Computing\/} (1991), pp.~164--174.
\vspace{-1mm}

\bibitem{Babai+JCSS95}
{\sc Babai, L., Cooperman, G., Finkelstein, L., Luks, E.~M., and Seress,
  {\'A}.}
\newblock Fast {M}onte {C}arlo algorithms for permutation groups.
\newblock {\em Journal of Computer and System Sciences 50}, 2 (1995), 296--308.
\vspace{-1mm}

\bibitem{Babai+FOCS84}
{\sc Babai, L., and Szemer{\'e}di, E.}
\newblock On the complexity of matrix group problems {I}.
\newblock In {\em Proceedings of the 25th Annual Symposium on Foundations of
  Computer Science\/} (1984), pp.~229--240.
\vspace{-1mm}

\bibitem{Bacon+FOCS05}
{\sc Bacon, D., Childs, A.~M., and van Dam, W.}
\newblock From optimal measurement to efficient quantum algorithms for the
  hidden subgroup problem over semidirect product groups.
\newblock In {\em Proceedings of the 46th Annual IEEE Symposium on Foundations
  of Computer Science\/} (2005), pp.~469--478.
\vspace{-1mm}

\bibitem{Buchmann+05}
{\sc Buchmann, J., and Schmidt, A.}
\newblock Computing the structure of a finite abelian group.
\newblock {\em Mathematics of Computation 74}, 252 (2005), 2017--2026.
\vspace{-1mm}

\bibitem{Cantor+81}
{\sc Cantor, D., and Zassenhaus, H.}
\newblock A new algorithm for factoring polynomials over finite fields.
\newblock {\em Mathematics of Computation 36\/} (1981), 587--592.
\vspace{-1mm}

\bibitem{Cheung+01}
{\sc Cheung, K., and Mosca, M.}
\newblock Decomposing finite abelian groups.
\newblock {\em Quantum Information and Computation 1}, 3 (2001), 26--32.
\vspace{-1mm}

\bibitem{Chou+SICOMP82}
{\sc Chou, T.-W.~J., and Collins, G.~E.}
\newblock Algorithms for the solution of systems of linear diophantine
  equations.
\newblock {\em SIAM Journal on Computing 11}, 4 (1982), 687--708.
\vspace{-1mm}

\bibitem{cooperman+97}
{\sc Cooperman, G., Finkelstein, L., and Linton, S.}
\newblock Recognizing {$GL_n(2)$} in non-standard representation.
\newblock In {\em Groups and Computation II, Proceedings of a SIMACS
  Workshop\/} (1997), pp.~85--100.
\vspace{-1mm}

\bibitem{Dummit+04}
{\sc Dummit, D.~S., and Foote, R.~M.}
\newblock {\em Abtract algebra}.
\newblock John Wiley and Sons, 2004.
\vspace{-1mm}

\bibitem{Ettinger+00}
{\sc Ettinger, M., and H{\o}yer, P.}
\newblock On quantum algorithms for noncommutative hidden subgroups.
\newblock {\em Advances in Applied Mathematics 25}, 3 (2000), 239--251.
\vspace{-1mm}

\bibitem{Friedl+STOC03}
{\sc Friedl, K., Ivanyos, G., Magniez, F., Santha, M., and Sen, P.}
\newblock Hidden translation and orbit coset in quantum computing.
\newblock In {\em Proceedings of the 35th Annual ACM Symposium on Theory of
  Computing\/} (2003), pp.~1--9.
\vspace{-1mm}

\bibitem{Garzon+JCSS91}
{\sc Garzon, M.~H., and Zalcstein, Y.}
\newblock On isomorphism testing of a class of 2-nilpotent groups.
\newblock {\em Journal of Computer and System Sciences 42}, 2 (1991), 237--248.
\vspace{-1mm}

\bibitem{Inui+QIC07}
{\sc Inui, Y., and {Le Gall}, F.}
\newblock Efficient quantum algorithms for the hidden subgroup problem over a
  class of semi-direct product groups.
\newblock {\em Quantum Information and Computation 7}, 5\&6 (2007), 559--570.
\vspace{-1mm}

\bibitem{Ivanyos+03}
{\sc Ivanyos, G., Magniez, F., and Santha, M.}
\newblock Efficient quantum algorithms for some instances of the non-abelian
  hidden subgroup problem.
\newblock {\em International Journal of Foundations of Computer Science 14}, 5
  (2003), 723--740.
\vspace{-1mm}

\bibitem{Kannan+SICOMP79}
{\sc Kannan, R., and Bachem, A.}
\newblock Polynomial algorithms for computing the {S}mith and {H}ermite normal
  forms of an integer matrix.
\newblock {\em SIAM Journal on Computing 8}, 4 (1979), 499--507.
\vspace{-1mm}

\bibitem{Kantor+01}
{\sc Kantor, W., and Seress, {\'A}.}
\newblock {\em Black box classical groups}.
\newblock American Mathematical Society, 2001.
\vspace{-1mm}

\bibitem{KavithaJCSS07}
{\sc Kavitha, T.}
\newblock Linear time algorithms for abelian group isomorphism and related
  problems.
\newblock {\em Journal of Computer and System Sciences 73}, 6 (2007), 986--996.
\vspace{-1mm}

\bibitem{Kitaev95}
{\sc Kitaev, A.~Y.}
\newblock Quantum measurements and the abelian stabilizer problem.
\newblock arXiv.org e-Print archive, arXiv:quant-ph/9511026, 1995.
\vspace{-1mm}

\bibitem{Kobler+93}
{\sc K{\"o}bler, J., Tor{\'a}n, J., and Sch{\"o}ning, U.}
\newblock {\em The graph isomorphism problem: its structural complexity}.
\newblock Birkh{\"a}user, 1993.
\vspace{-1mm}

\bibitem{LeGallSTACS09}
{\sc {Le Gall}, F.}
\newblock Efficient isomorphism testing for a class of group extensions.
\newblock In {\em Proceedings of the 26th International Symposium on
  Theoretical Aspects of Computer Science\/} (2009), pp.~625--636.
  Full version available at http://arxiv.org/abs/0812.2298.
\vspace{-1mm}

\bibitem{Lidl+08}
{\sc Lidl, R., and Niederreiter, H.}
\newblock {\em Finite fields}.
\newblock Cambridge University Press, 2008.
\vspace{-1mm}

\bibitem{Lipton+76}
{\sc Lipton, R.~J., Snyder, L., and Zalcstein, Y.}
\newblock The complexity of word and isomorphism problems for finite groups.
\newblock Tech. rep., John Hopkins, 1976.
\vspace{-1mm}

\bibitem{McKenzie+SICOMP87}
{\sc McKenzie, P., and Cook, S.~A.}
\newblock The parallel complexity of abelian permutation group problems.
\newblock {\em SIAM Journal on Computing 16}, 5 (1987), 880--909.
\vspace{-1mm}

\bibitem{MillerSTOC78}
{\sc Miller, G.}
\newblock On the $n^{\log n}$ isomorphism technique.
\newblock In {\em Proceedings of the 10th Annual ACM Symposium on Theory of
  Computing\/} (1978), pp.~51--58.
\vspace{-1mm}

\bibitem{Moore+SODA04}
{\sc Moore, C., Rockmore, D.~N., Russell, A., and Schulman, L.~J.}
\newblock The power of basis selection in fourier sampling: hidden subgroup
  problems in affine groups.
\newblock In {\em Proceedings of the 15th Annual ACM-SIAM Symposium on
  Discrete Algorithms\/} (2004), pp.~1113--1122.
\vspace{-1mm}

\bibitem{Mosca99}
{\sc Mosca, M.}
\newblock {\em Quantum Computer Algorithms}.
\newblock PhD thesis, Oxford university, 1999.
\vspace{-1mm}

\bibitem{Ranum07}
{\sc Ranum, A.}
\newblock The group of classes of congruent matrices with application to the
  group of isomorphisms.
\newblock {\em Transactions of the American Mathematical Society 8}, 1 (1907),
  71--91.
\vspace{-1mm}

\bibitem{ShorSICOMP97}
{\sc Shor, P.~W.}
\newblock Polynomial-time algorithms for prime factorization and discrete
  logarithms on a quantum computer.
\newblock {\em SIAM Journal on Computing 26}, 5 (1997), 1484--1509.
\vspace{-1mm}

\bibitem{StorjohannFOCS01}
{\sc Storjohann, A.}
\newblock Deterministic computation of the {F}robenius form.
\newblock In {\em Proceedings of the 42nd Annual Symposium on Foundations of
  Computer Science\/} (2001), pp.~368--377.
\vspace{-1mm}

\bibitem{Vikas96}
{\sc Vikas, N.}
\newblock An $O(n)$ algorithm for Abelian $p$-group isomorphism and an $O(n\: \log
  n)$ algorithm for Abelian group isomorphism.
\newblock {\em Journal of Computer and System Sciences 53}, 1 (1996), $\:\:\:\:\:\:\:\:\:\:$1--9.
\vspace{-1mm}

\bibitem{WatrousFOCS00}
{\sc Watrous, J.}
\newblock Succinct quantum proofs for properties of finite groups.
\newblock In {\em Proceedings of the 41st Annual Symposium on Foundations of
  Computer Science\/} (2000), pp.~537--546.
\vspace{-1mm}

\bibitem{WatrousSTOC01}
{\sc Watrous, J.}
\newblock Quantum algorithms for solvable groups.
\newblock In {\em Proceedings of the 33rd Annual ACM Symposium on Theory of
  Computing\/} (2001), pp.~60--67.

\end{thebibliography}
\end{document}